%% file: main.tex
\documentclass[letterpaper, 10 pt, conference]{ieeeconf}  

\IEEEoverridecommandlockouts                              

\title{\LARGE \bf 
R3R: Decentralized Multi-Agent Collision Avoidance with Infinite-Horizon Safety
}

\author{Thomas Marshall Vielmetti$^{1}$, Devansh R. Agrawal$^{2}$, and Dimitra Panagou$^{3}$
\thanks{The authors would like to acknowledge the support of the National Science Foundation (NSF) under grant no. 2137195.}
\thanks{$^{1}$Thomas Marshall Vielmetti is with the Department of Electrical and Computer Engineering,
        University of Michigan, Ann Arbor, MI 48109, USA
        {\tt\small mvielmet@umich.edu}}%
\thanks{$^{2}$Devansh R. Agrawal is with the Department of Aerospace Engineering,
        University of Michigan, Ann Arbor, MI 48109, USA
        {\tt\small devansh@umich.edu}}%
\thanks{$^{3}$Dimitra Panagou is with the Department of Robotics and Department of Aerospace Engineering,
        University of Michigan, Ann Arbor, MI 48109, USA
        {\tt\small dpanagou@umich.edu}}%
}

\input{preamble/preamble}

\RestyleAlgo{ruled}

\usepackage{comment}

\usepackage{cleveref}
\crefname{assumption}{assumption}{assumptions}
\Crefname{assumption}{Assumption}{Assumptions}

\crefname{definition}{Def.}{Defs.}
\crefname{figure}{Fig.}{Figs.}

\theoremstyle{theorem}

\crefname{protocol}{Protocol}{protocols}
\Crefname{protocol}{Protocol}{protocols}

\SetAlgoSkip{}

\begin{document}

\maketitle
\thispagestyle{empty}
\pagestyle{empty}

\begin{abstract}
Existing decentralized methods for multi-agent motion planning lack formal, infinite-horizon safety guarantees, especially for communication-constrained systems.
We present R3R which, to our knowledge, is the first decentralized and asynchronous framework for multi-agent motion planning under range-limited communication constraints with infinite-horizon safety guarantees for systems of nonlinear agents.
R3R's novelty lies in combining our \gatekeeper{} safety framework with a geometric constraint termed R-Boundedness, which together establish a formal link between an agent's communication radius and its ability to plan safely.
We constrain trajectories to lie within a fixed planning radius, determined by a function of the agent's communication radius.
This enables trajectories to be certified as provably safe for all time using only local information.
Our algorithm is fully asynchronous, and ensures the forward invariance of these guarantees even in time-varying networks where agents asynchronously join and replan.
We evaluate our approach in simulations of up to 128 Dubins vehicles, validating our theoretical safety guarantees in dense, obstacle-rich scenarios.
We further show that R3R's computational complexity scales with local agent density rather than problem size, providing a practical solution for scalable and provably safe multi-agent systems\footnote{Code available at \url{github.com/MarshallVielmetti/r3r_decentralized_gatekeeper}}.
\end{abstract}

\section{Introduction}\label{sec:introduction}
\input{sections/1_introduction.tex}

\section{Preliminaries and Problem Statement}\label{sec:prelims}

\emph{Notation}: $\R, \Rnonneg, \Rplus$ are the sets of reals, non-negative reals, and positive reals. Let $\Bcal_R(x) = \{ y : \norm{x - y}_2 \leq R\}$ be the closed ball of radius $R \geq 0$ centered at $x \in \R^n$. $\Acal(t) \subset \mathbb{N}$ denotes the set of all agents at time $t$, where $N = |\Acal(t)|$ is the total number of agents in the network at that time. 

\subsection{Preliminaries}
\noindent Consider a dynamical system
\begin{equation}\label{eq:system}
    \dot x = f(t, x, u),
\end{equation}
where $x\in\Xcal\subseteq\R^n$ is the state and $u \in \Ucal \subseteq \R^m$ is the control input. The function $f : \Rnonneg \times \Xcal \times \Ucal \to \R^n$ is piecewise continuous in $t$ and locally Lipschitz in $x$ and $u$. Agents operate in a workspace $\Wcal \subset \R^d$, where $d \in \set{2, 3}$. Each agent's state $x$ contains a position component, $x = \bmat{p^\top, \bar{x}^\top}^\top$, where $p \in \Wcal$ is the position and $\bar{x} \in \R^{n - d}$ represents the remaining states. Define the projection $\Pi : \Xcal \to \Wcal$ by $\Pi(x) = p$.

\begin{definition}[Agent]\label{def:agent}
    An agent represents a robot or other actor that interacts with the system of other agents.  An agent $i \in \Acal$ has the following properties:
    \begin{enumerate}
        \item State $x_i$ which evolves through a system satisfying~\eqref{eq:system}.
        \item Inter-agent avoidance distance of $\delta \in \Rplus$, meaning no agent may come within a distance strictly less than $\delta$ of agent $i$.
        \item Communication radius $\Rcomm \in \Rplus$.
    \end{enumerate}
\end{definition}
\begin{assumption}\label{assumption:homogeneous}
    All agents have a common $\delta$ and $\Rcomm$.
\end{assumption}

\begin{definition}[Trajectory]
    A \textbf{trajectory}\footnote{The time-interval $T$ can be a finite horizon, e.g., $T = [t_0, t_1]$ or an infinite horizon, e.g., $T = [t_0, \infty)$.} is the tuple $\Tcal = (T, \x, \u)$, where $T \subset \Rnonneg$ is a time-interval, $\x: T \to \Xcal$ is the state-trajectory, and $\u: T \to \Ucal$ is a control-trajectory, satisfying:
    \eqn{
        \dot \x(t) = f(t, \x(t), \u(t)), \quad \forall{} t \in T.
    }
\end{definition}

\begin{definition}[Communication Network]
    \label{def:undirected_graph}
    Let $\Gcal (t) = (\Acal(t) , E(t))$ be an undirected, time-varying graph, with nodes representing agents and edges representing communication links, such that:
    \eqn{
        E(t) = \set{(i,j) \in \Acal(t) : \norm{\p_i(t) - \p_j(t)}_2 \leq \Rcomm}.
    }
\end{definition}

\begin{assumption}\label{assumption:instantaneous_comms}
    Communication is instantaneous and range-limited, but not bandwidth-limited.
    That is, if two agents $i, j$ are neighbors at some time $t$, they can communicate an arbitrary amount of information instantaneously. 
\end{assumption}

\noindent Let
\eqn{
    \Ncal_i(t) = \{ j \in \Acal(t) : (i,j) \in E(t) \}
}
denote the neighbors of agent $i$ at time $t$.

\subsection{Problem Statement}
We now formally state the problem of provably safe multi-agent motion planning under range-limited communication.
\begin{problem}\label{problem:one}
    Let $\Acal(t) \subset \mathbb{N} $ denote a time-varying set of agents satisfying \cref{def:agent}, that communicate over a connectivity graph $\Gcal(t)$ as defined by \cref{def:undirected_graph}.
    Let $\Wcal_\text{obs} \subset \Wcal$ be the subset of the domain occupied by static obstacles, and $\Scal \subset \Wcal \setminus \Wcal_{\text{obs}}$ be the safe set, known to all agents.
   Agents attempt to navigate to their goal state $g_i \in \Xcal$,
    \eqn{
         \lim_{t \to \infty} x_i(t) = g_i, \quad \forall i \in \Acal,
    }
    but must satisfy collision avoidance and safety constraints:
    \eqn{
        \p_i(t) \in \Scal, && \quad \forall t \in \Rnonneg,\\
        \norm{\p_i(t) - \p_j(t)}_2 \geq \delta, &&\forall j \in \Acal(t) \setminus \set{i}, \forall t \in \Rnonneg.
    }
    We assume $\Acal(t_0)$ is initially empty, and agents are removed from consideration upon reaching their goal state. 
\end{problem}
To solve \cref{problem:one}, we propose R3R, a decentralized framework for multi-agent motion planning under communication constraints.

\section{Proposed Approach and Solution}\label{sec:approach}
We construct our proposed approach in three stages. In \cref{subsec:gatekeeper_preliminaries}, we introduce our prior work for single-agent safety \gatekeeper{}~\cite{agrawalGatekeeperOnlineSafety2024}. In \cref{subsec:r_boundedness}, we introduce a novel geometric constraint relating agents' ability to communicate to their ability to plan safely, which we show enables system-wide safety certification using only local information. Finally, in \cref{subsec:decentralized_safety_protocol}, we integrate these two systems into a provably safe, decentralized algorithm to solve \cref{problem:one}. An overview of the proposed architecture is shown in \cref{fig:r3r_principle_diagram}.

\begin{figure}
    \centering
    \includegraphics[width=0.35\textwidth]{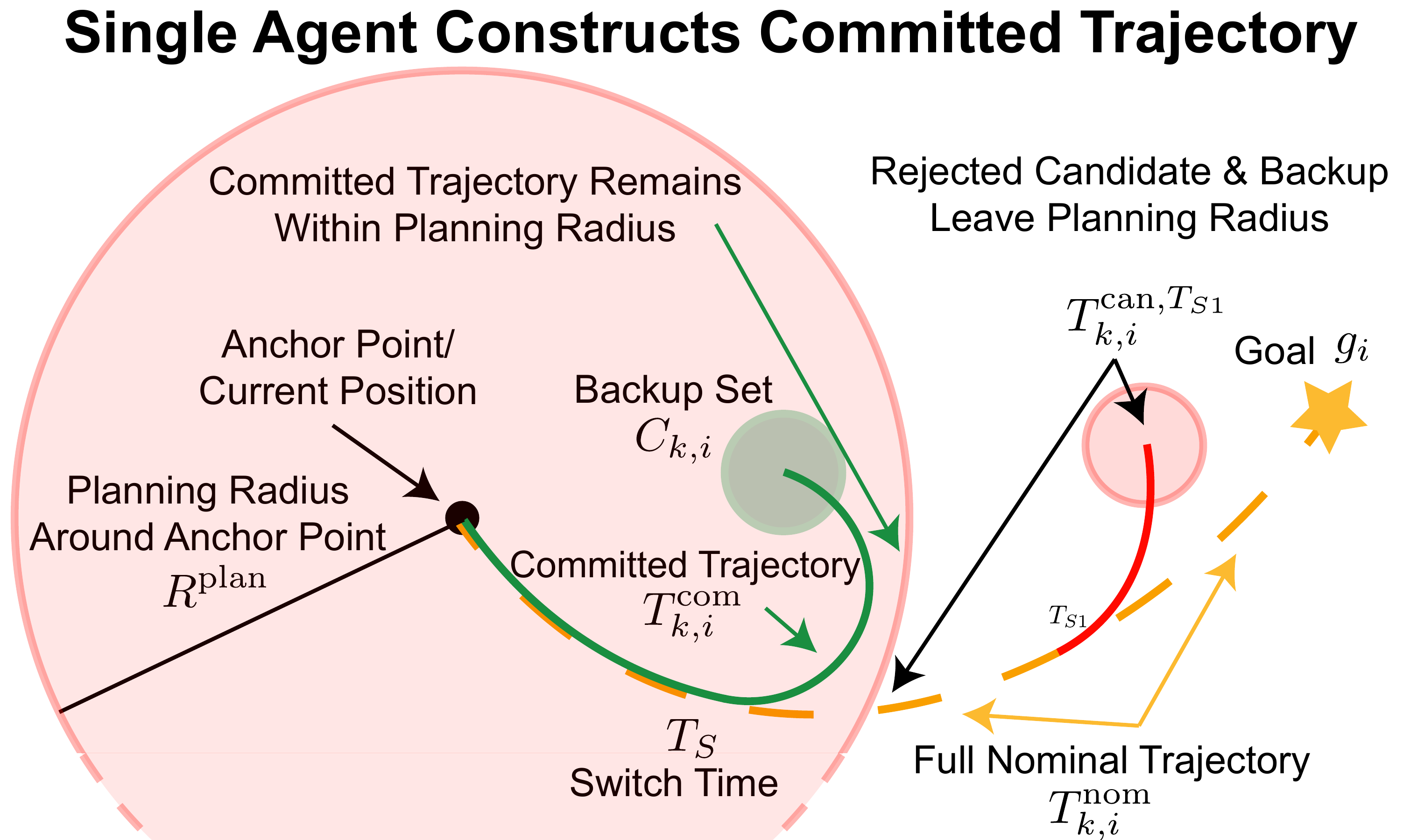}
    \vspace{-0.25cm}
    \caption{\textbf{Constructing a Safe, Committed Trajectory.} An agent at the anchor point plans a nominal trajectory (orange) toward its goal. To ensure safety, it must find a candidate trajectory that remains entirely within its planning radius $\Rplan$ (red circle). An unsafe candidate that exits this radius (red) is rejected. A safe candidate that stays within the radius (green) is verified and becomes the committed trajectory for the agent to follow.}
    \label{fig:r3r_principle_diagram}
\end{figure}

\begin{figure*}[!ht]
    \centering
    \includegraphics[width=1.0\textwidth]{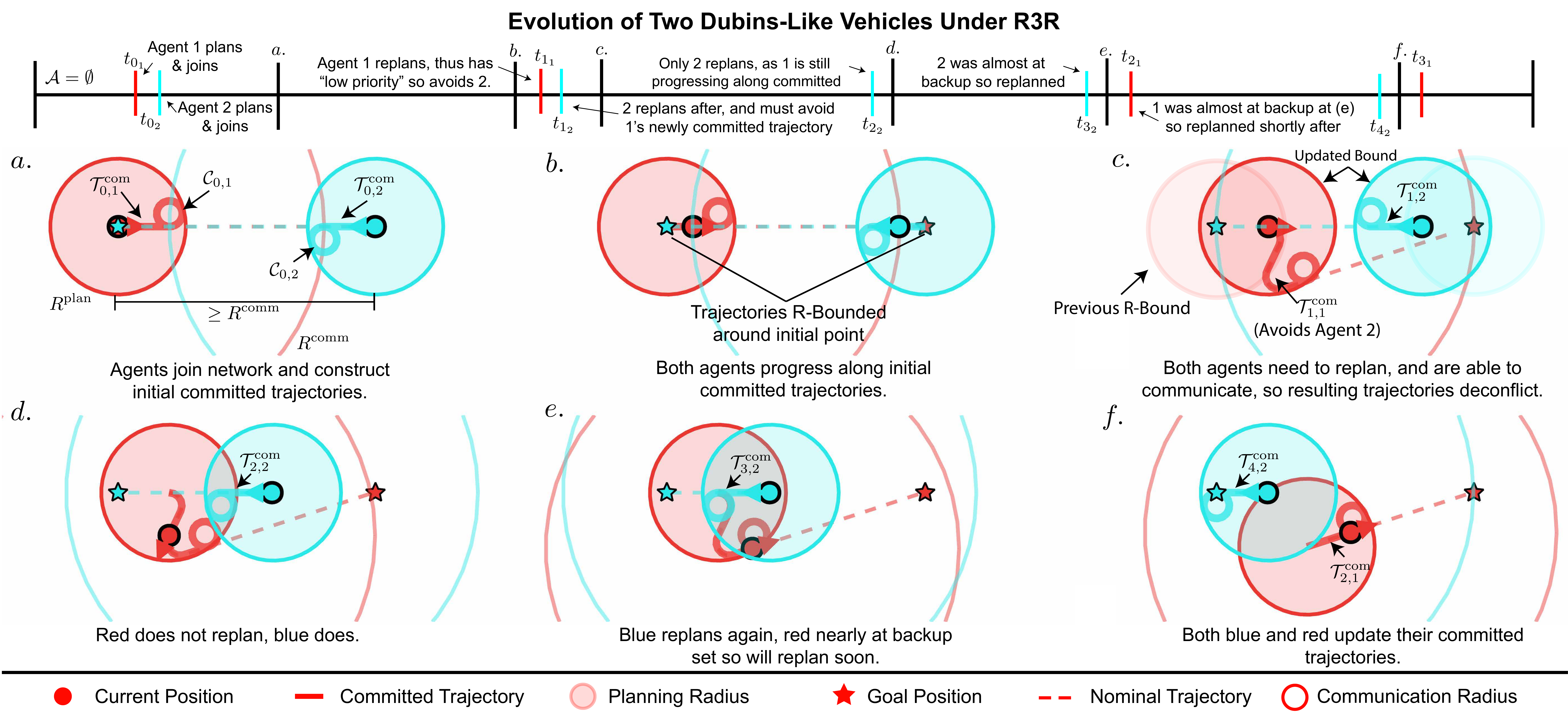}
    \caption{
    \textbf{Demonstration of R3R over time for two Dubins agents}. The timeline indicates when each agent (Red: Agent 1, Blue: Agent 2) plans or replans.
    \emph{Committed trajectories} (solid lines) are certified to be safe for all time, and remain within a planning radius $\Rplan$ (filled circles), around where they are constructed. 
    Agents first plan goal-oriented nominal trajectories (dashed lines). 
    A portion of this nominal is combined with a backup trajectory, to form a candidate. If this candidate is deemed safe, it becomes a committed trajectory. 
    At times $t_{0,1} \neq t_{0,2}$ agents $1$ and $2$ wish to join the network, then independently construct a valid committed trajectory, which is then simultaneously committed and broadcast as the agent enters the network.
    \textbf{(a.)} The initial committed trajectories of both agents remain within the planning radius.
    \textbf{(b.)} Agents continue to track their initial committed trajectories, and enter communication range. 
    \textbf{(c.)} Agent $1$ replanned first, so constructs a new candidate trajectory, which it certifies as safe with respect to the current committed trajectory of agent $2$. At time $t_{1_2} > t_{1_1}$, agent $2$ chooses to replan, and must now avoid the recently committed trajectory of agent $1$, which reflects the asynchronous nature of the algorithm; at subsequent replanning steps, agents may not invalidate previously certified decisions of neighbors. 
    \textbf{(d-f.)} The system continues to evolve, and agents continue to asynchronously replan.
    }
    \label{fig:r3r_evolution}
    \vspace{-0.5cm}
\end{figure*}

\subsection{\gatekeeper{} Preliminaries}\label{subsec:gatekeeper_preliminaries}

R3R builds upon our prior work for single-agent safety \gatekeeper{} \cite{agrawalGatekeeperOnlineSafety2024}, which is executed independently by each agent $i \in \Acal$ at discrete, asynchronous iterations indexed $k_i \in \naturals$ at times $t_{k_i}$, such that $t_{k_i} < t_{k_i + 1}$.
We ensure agents always track a \emph{committed trajectory} that is known to be safe for all future time.
At each iteration, agents attempt to construct a new, provably safe, \emph{committed trajectory}, by properly appending a goal-oriented \emph{nominal trajectory} with a safe, infinite-horizon \emph{backup trajectory}.
This becomes a \emph{candidate trajectory}, which must be validated as safe prior to being committed. 

\begin{definition}[Nominal Trajectory]\label{def:nominal_trajectory}
    Given state $\x_i(t_{k_i}) \in \Xcal$, at some time $t_{k_i} \in \Rnonneg$, the \textbf{nominal trajectory} of agent $i$, denoted
    \eqn{
        (T_{k_i,i}^\nom, \x_{k_i,i}^\nom, \u_{k_i,i}^\nom),
    }
    is defined over the horizon $T_{k_i,i}^\nom = [t_{k_i}, t_{k_i} + \Delta t_H]$, where $\Delta t_H \in \Rplus$ is the planning horizon.
\end{definition}
\noindent A nominal trajectory is not guaranteed to satisfy constraints nor provide any guarantees beyond $\Delta t_H$.

\begin{definition}[Backup Set]\label{def:backup_set}
    A set $\Ccal \subset \Xcal$ is a \textbf{backup set} if $\Pi(\Ccal) \subset \Scal$, and
    there exists a controller $\pi^B : \R \times \Xcal \to \Ucal$ such that $\forall t_0 \in \Rnonneg$, the closed-loop system $\dot x = f(t, x, \pi^B(t, x))$ satisfies:
    \eqn{
        \x(t_0) \in \Ccal \implies \x(t) \in \Ccal, \forall t \geq t_0.
    }
\end{definition}

\begin{definition}[Backup Trajectory]\label{def:backup_trajectory}
    Let $\Delta t_B \in \Rplus$. For any $t_0 \in \Rplus$, $x_0 \in \Xcal$, and backup set $\Ccal$, a trajectory $(T^\back, \x^\back, \u^\back)$, defined on $T^\back = [t_0, \infty)$ is a \textbf{backup trajectory} if:
    \eqn{
        \x^\back(t_0) &= x_0,\\
        \p^\back(t) &\in \Scal, \; \forall t \in [t_0, t_0 + \Delta t_B]\\
        \x^\back(t_0 + \Delta t_B) &\in \Ccal,
    }
    and for all $t \geq t_0 + \Delta t_B$, the system evolves under the backup controller
    \eqn{
        \u^\back(t) = \pi^B(t, \x^\back(t)).
    }
\end{definition}

\noindent A backup trajectory takes a system from a state $x_0\in\Xcal$ at time $t_0$ such that at time $t_0 + \Delta t_B$, the system lies in a backup set. 
The system then executes the backup controller for all future time, thus remaining in the backup set.
If the finite horizon trajectory defined over $[t_0, t_0 + \Delta t_B]$ safely takes the agent into the backup set, then by the properties of the backup set we guarantee the agent will remain in $\Scal$ over $[t_0, \infty)$. 

\begin{definition}[Candidate Trajectory]\label{def:candidate_trajectory}
    At time $t_{k_i} \in \Rnonneg$, let agent $i \in \Acal(t_{k_i})$ be at state $\x_i(t_{k_i}) \in \Xcal$. 
    Let the agent's nominal trajectory be $\Tcal_{k_i,i}^\nom = ([t_{k_i}, t_{k_i} + \Delta t_H], \x_{{k_i},i}^\nom, \u_{k_i,i}^\nom)$.
    A \textbf{candidate trajectory} with switch time $\Delta t_S \in [0, \Delta t_H]$, denoted $\Tcal_{k_i,i}^{\can,\Delta t_S} = ([t_{k_i}, \infty), \x_{k_i,i}^\can, \u_{k_i,i}^\can)$, is defined by:
    \eqn{
        \u_{k_i,i}^\can(\tau) = 
        \begin{cases}
            \u_{k_i,i}^\nom(\tau) & \text{if } \tau \in [t_{k_i}, t_{k_i} + \Delta t_S)\\
            \u_{k_i,i}^\back(\tau) & \text{if } \tau \geq t_{k_i} + \Delta t_S
        \end{cases},
    }
    where $([t_{k_i} + \Delta t_S, \infty), \x_{k_i,i}^\back, \u_{k_i,i}^\back)$ is a backup trajectory as defined in \cref{def:backup_trajectory}, and $\x_{k_i,i}^\can(t)$ is the corresponding state trajectory.
\end{definition}

\subsection{$R$-Boundedness}\label{subsec:r_boundedness}
Our core challenge is to enable system-wide certification of candidate trajectories using local information.
To accomplish this, we introduce a geometric constraint called $R$-Boundedness, which restricts an agent to planning within a known, finite volume.

\begin{definition}[$R$-Bounded]
    A trajectory $\Tcal = ([t_0, \infty), \x(t), \u(t))$ is \textbf{R-Bounded} if:
    \eqn{
        \norm{\p(t) - \p(t_0)}_2 \leq R,\; \forall t \geq t_0.
    }
\end{definition}
\noindent That is, a trajectory is $R$-Bounded if it remains within a ball of radius $R$ around its starting position, its \emph{anchor point}, for all future time.
\begin{lemma}[Collision of Bounded Trajectories]
    \label{lemma:collision_r_bounded}
    Consider two $R$-Bounded trajectories $([t_1, \infty), \x_1, \u_1),\; ([t_2, \infty), \x_2, \u_2)$, and suppose the system has an inter-agent collision radius of $\delta < R$. Without loss of generality, let $t_1 \leq t_2$.
    If
    \begin{equation}
        \norm{\p_1(t_1) - \p_2(t_2)}_2 \geq 2R + \delta,
    \end{equation}
    then
    \eqn{
        \norm{\p_1(t) - \p_2(t)}_2 \geq \delta, \forall t \geq t_2.
    }
\end{lemma}
\begin{proof}
    Since trajectories are $R$-Bounded, $\p_1(t) \in B_1 = \Bcal_R(\p_1(t_1)), \; \p_2(t) \in B_2 = \Bcal_R(\p_2(t_2))$, $\forall t \geq t_2$.
    Then, 
    \neqn{
        \min_{t\geq t_2} \norm{\p_1(t) - \p_2(t)}_2 &\geq \min_{p_1 \in B_1, \; p_2 \in B_2} \norm{p_1 - p_2}_2\\
        &= \norm{\p_1(t_1) - \p_2(t_2)}_2 - 2R \geq \delta.
    }
    Hence, the trajectories are collision-free. 
\end{proof}

While \cref{lemma:collision_r_bounded} establishes safety based on the trajectory anchor points, we must translate this into a condition based on the agents' current positions at the moment of replanning.

\begin{lemma}[Sufficient Communication Radius]\label{lemma:time_of_committed}
    Consider two $R$-Bounded trajectories $([t_1, \infty), \x_1, \u_1),\; ([t_2, \infty), \x_2, \u_2)$, with starting times $t_1 < t_2$ for some radius $R$. If at some $t \geq t_2$ the agents collide,
    \eqn{
        \norm{\p_2(t) - \p_1(t)}_2 < \delta,
    }
    then they must have been within $3R + \delta$ at time $t_2$, i.e.,
    \eqn{
        \norm{\p_2(t_2) - \p_1(t_2)}_2 < 3R + \delta.
    } 
    Thus, if trajectories are bounded by a radius $\Rplan$ satisfying
    \eqn{\label{eqn:r3r}
        \Rcomm = 3\Rplan + \delta,
    }
    which we term the R3R condition, then agents $1$ and $2$ must be neighbors at $t_2$, i.e., $(1, 2) \in E(t_2)$.
\end{lemma}
\begin{proof}
    The contrapositive of \cref{lemma:collision_r_bounded} gives
    \eqn{
        \norm{\p_2(t) - \p_1(t)}_2 < \delta \Rightarrow \norm{\p_2(t_2) - \p_1(t_1)}_2 < 2R + \delta.
    }
    By the definition of $R$-Boundedness, we have
    \eqn{
        \norm{\p_1(t_2) - \p_1(t_1)}_2 \leq R,
    }
    so by the triangle inequality,
    \neqn{
        \norm{\p_2(t_2) - \p_1(t_2)}_2
            & \leq \norm{\p_2(t_2) - \p_1(t_1)}_2 + \norm{\p_1(t_2) - \p_1(t_1)}_2 \\
            &< 2R + \delta + R = 3R + \delta.
    }
    By letting $R = \Rplan$ and applying the condition $\Rcomm = 3\Rplan + \delta$, we have
    \neqn{
        \norm{\p_2(t_2) - \p_1(t_2)}_2 &\le 3\Rplan + \delta,\\
                                       &= \Rcomm.
    }
    And therefore, by \cref{def:undirected_graph}, $(1, 2) \in E(t_2)$.
\end{proof}
This extra $R$ relative to \cref{lemma:collision_r_bounded} arises from the fact that agent $1$ may move up to a distance $R$ between $t_1$ and $t_2$, and results in Eq.~\eqref{eqn:r3r}, the R3R condition. This informs a constraint on trajectories which, if satisfied by all agents in the system, enables global safety certification with local information. 

\subsection{Decentralized Safety Protocol}\label{subsec:decentralized_safety_protocol}
We now integrate this constraint with \gatekeeper{} to establish a formally safe decentralized protocol.
We use \gatekeeper{} as the building block for safety, and leverage $R$-Boundedness to enable local safety verification of candidate trajectories.

When an agent $i$ joins the network, which we term iteration $0_i$, the agent is required to have a valid candidate. At later iterations $k_i > 0$, agent $i$ will:
\begin{enumerate}
    \item Receive the current committed trajectories $\Tcal_{k_j,j}^\com$ of all neighbors $j \in \Ncal_i(t_{k_i})$.
    \item Plan a nominal trajectory $\Tcal_{k_i, i}^\nom$, attempting to avoid static obstacles and the committed trajectories of neighbors.
    \item Construct candidates using \cref{alg:attempt_replan} until one is found which satisfies the validity conditions of Def.~\ref{def:valid}, or the algorithm returns.
    \item If a valid candidate is found, commit the trajectory, otherwise continue along $\Tcal_{k_i-1,i}^\com$.
\end{enumerate}

\begin{remark}\label{remark:triggered_iterations}
No conditions are placed on the frequency of replanning iterations which can be triggered independently and asynchronously by agents based on internal logic. 
\end{remark}

We now formalize the conditions under which a candidate trajectory can be certified as safe.
\begin{definition}[Valid Candidate Trajectory] \label{def:valid}
At time $t_{k_i}$, let agent $i$ hold a previously committed trajectory $\Tcal^{\com}_{k_i-1, i}$. A candidate trajectory $\Tcal^{\can, \Delta t_S}_{k_i, i}$ is \textit{valid} if it satisfies:
\begin{enumerate}
    \item \textbf{Safe Set:} $\p_{k_i, i}^{\can, \Delta t_S}(t) \in \Scal, \forall t \in [t_{k_i}, t_{k_i, B}]$.
    \item \textbf{Backup Reachability:} 
    \eqn{
        \x_{k_i, i}^{\can, \Delta t_S}(t_{k_i, B}) \in \Ccal_{k_i, i},
    }
    where $t_{k_i,B} = t_{k_i} + \Delta t_S + \Delta t_B$.
    \item \textbf{Backup Separation:}
    \eqn{\label{eqn:backup_separation}
        \inf_{\substack{x_1\in \Ccal_{k_i,i} \\ x_2 \in \Ccal_{k_j,j}}} \norm{\Pi(x_1) - \Pi(x_2)}_2 \ge \delta, \; \forall j \in \Ncal_{i}(t_{k_i}).
    }
    \item \textbf{R-Boundedness:} 
    \neqn{
        \norm{\p_{k_i, i}^{\can, \Delta t_S}(t) - \p_i(t_{k_i})}_2 &\le \Rplan, \forall t \in [t_{k_i}, t_{k_i, B}],\\
        \sup_{x \in \Ccal_{k_i,i}} \norm{\Pi(x) - \p_i(t_{k_i})}_2 &\le \Rplan. 
    } 
    \item \textbf{Collision Avoidance:} $\forall t \in [t_{k_i}, \max\set{t_{k_i,B}, t_{k_j,B}}]$,
    \eqn{\label{eqn:collision_avoidance}
        \norm{\p_{k_i, i}^{\can, \Delta t_S}(t) - \p_{k_j,j}^\com(t)}_2 \ge \delta,\; \forall j \in \Ncal_i(t_{k_i}).
    }
\end{enumerate}
\end{definition}
\begin{remark}
The conditions given by \eqref{eqn:backup_separation} and \eqref{eqn:collision_avoidance} together provide an infinite-horizon safety guarantee.
Condition \eqref{eqn:collision_avoidance} ensures finite-horizon collision avoidance while the agents are actively traversing their trajectories, and condition \eqref{eqn:backup_separation} ensures that once both agents reach their respective backup sets, they will remain permanently separated for all future time.
\end{remark}

We can now establish the central pillar of our safety argument: if every agent executes a trajectory that meets these local validity conditions, then the entire multi-agent system is provably safe for all future time. This is because the $R$-Boundedness condition geometrically precludes collisions with agents outside of the local communication radius.

\begin{theorem}[R3R - Safety of Valid Trajectories]\label{theorem:no_update_safety}
    If all agents $i \in \Acal(t_k)$ choose $\Rplan$ satisfying the R3R condition \eqref{eqn:r3r}, and execute trajectories $\Tcal_{k_i,i}^\com$ that were valid by Def.~\ref{def:valid} when committed at times $t_{k_i} < t_k$, then $\forall t \geq t_k$, agents executing these committed trajectories remain safe: 
    \eqn{
        \norm{\p_i(t) - \p_j(t)}_2 \geq \delta, & \quad \forall i,j \in \Acal, i\neq j,\\
        \p_i(t) \in \Scal, & \quad \forall i \in \Acal.
    }
\end{theorem}
\begin{proof}
    \emph{Collision avoidance.}
    Seeking contradiction, assume that at some $t \geq t_k$, distinct agents $i,j \in \Acal$ collide, i.e., $\exists t\geq t_k$ such that $\norm{\p_i(t) - \p_j(t)}_2 < \delta$. 
    Without loss of generality, let $t_{k_i} < t_{k_j}$. 
    Both trajectories are valid by \cref{def:valid}, and thus $R$-Bounded by $\Rplan$.
    Applying \cref{lemma:time_of_committed}, the condition $\norm{\p_i(t) - \p_j(t)}_2 < \delta$ implies that at time $t_{k_j}$, the agents' positions must have satisfied $\norm{\p_i(t_{k_j}) - \p_j(t_{k_j})}_2 < 3\Rplan + \delta$, which per the R3R condition \eqref{eqn:r3r} is precisely $\Rcomm$, thus $i \in \Ncal_j(t_{k_j})$.
    This is a contradiction, since by \cref{def:valid}, agent $j$'s trajectory $\p_j(t)$ being valid at $t_{k_j}$
    requires that it is collision-free with respect to the trajectories of all agents in $\Ncal_j(t_{k_j})$, and hence collision-free with respect to the trajectory $\p_i(t)$ of agent $i\in \Ncal_j(t_{k_j})$, $\forall t\geq t_{k_j}$.
    
    \emph{Safe set.} 
    A valid trajectory satisfies $\p_i(t) \in \Scal, \forall t \in [t_{k_i}, t_{k_i,B}]$, and $\x_i(t) \in \Ccal_i, \forall t \geq t_{k_i,B}$, which by \cref{def:backup_set}, implies $\Pi(\Ccal_i) \subset \Scal$, and thus $\p_i(t) \in \Scal, \forall t \geq t_{k_i}$.
\end{proof}

\Cref{theorem:no_update_safety} shows that the conditions outlined in Def.~\ref{def:valid}, which rely only on local information, still provide guarantees with respect to the rest of the agents in the network.
However, this theorem assumes the committed trajectories of each agent are held constant, and thus remain $R$-Bounded around the same anchor point.
In order to make progress, agents must continually update their committed trajectories.
We now introduce the necessary machinery to ensure all agents in the network always track valid trajectories, guaranteeing the conditions of \cref{theorem:no_update_safety} are satisfied after every iteration of \gatekeeper{}.

\begin{definition}[Committed Trajectory]\label{def:committed_trajectory}
    At the $k_i$-th iteration, let agent $i \in \Acal(t_{k_i})$ be at state $\x_i(t_{k_i}) \in \Xcal$, and $\Tcal_{k_i,i}^\can$ be the candidate trajectory.
    Define the set of valid switch times:
    \eqn{
        \Ical_{k_i} = \{ \Delta t_S \in [0, \Delta t_H] : \Tcal_{k_i,i}^\can  \text{ is valid.}\},
    } 
    where validity is defined by \cref{def:valid}.
    If $\Ical_{k_i} \neq \emptyset$, let $\Delta t_S^* = \max\;\Ical_{k_i}$, then $\Tcal_{k_i,i}^\com = \Tcal_{k_i,i}^{\can,\Delta t_S^*}$.
    Otherwise, let $\Tcal_{k_i,i}^\com = \Tcal_{k_i-1,i}^\com$ (i.e. continue tracking the previously committed trajectory).
\end{definition}

To handle the asynchronous nature of a decentralized system, where multiple agents may replan at different times, we must introduce an assumption to prevent conflicts where two neighbors invalidate each other's plans simultaneously.

\begin{assumption}\label{assumption:no_simultaneous_replanning}
    No two agents which can communicate update their committed trajectories simultaneously, i.e.,
    \neqn{
        t_k = t_{k_i} = t_{k_j} \implies (i,j) \notin E(t_k).
    } 
\end{assumption}
\begin{remark}
    This assumption is necessary to ensure replanning agents have accurate information about their neighbors \cite{vielmettiMultiAgentGatekeeperSafe2026}, preventing agents from simultaneously constructing mutually invalidating candidates.
    Under the conditions of \cref{assumption:instantaneous_comms}, this is a mild assumption, however we acknowledge handling the effects of non-ideal network communication as a direction for future work.
    An approach to relaxing this assumption is explored in \cite{kondoRobustMADERDecentralized2023}, using a check-recheck scheme to ensure agents have sufficiently up-to-date information about their neighbors.
\end{remark}

The algorithmic steps of checking validity by \cref{def:valid}, and constructing committed trajectories are shown in \cref{alg:attempt_replan,alg:update_state}.

\begin{algorithm}
\caption{AttemptReplan}
\label{alg:attempt_replan}
\KwIn{$x_i, g_i, \Ncal^{\com}, \Scal$}
\KwResult{\textit{success} (bool), $\Tcal^{\can,\Delta t_S}$}
\DontPrintSemicolon

$\Tcal^{\nom} \leftarrow \text{PlanNominal}(x_i, g_i, \Ncal^{\com}, \Scal)$\;

\For{$\Delta t_S = \Delta t_H, \; \Delta t_H - \Delta t, \dots, 0$}{

    $\Tcal^{\can,\Delta t_S} \leftarrow \text{ConstructCandidate}(\Tcal^{\nom}, \Delta t_S)$\;
    
    \If{$\text{IsValid}(\Tcal^{\can,\Delta t_S}, \Ncal^{\com}, \Scal)$}{
    
        \KwRet{True, $\Tcal^{\can,\Delta t_S}$}\;
    }
}
\KwRet{False, $\emptyset$}\;
\end{algorithm}
\begin{remark}
    At each iteration, when constructing candidate trajectories, agents must choose some backup set $\Ccal_{k_i}$. The same backup set can be used for each candidate, multiple can be evaluated, or its location can be optimized as part of the planning process. Agents include and transmit a representation of this backup set alongside their committed trajectories.
\end{remark}

\begin{algorithm}
    \caption{UpdateState}
    \label{alg:update_state}
    \SetKwComment{Comment}{// }{}
    \DontPrintSemicolon
    \Comment{Trigger: Agent $i$ join/replan.}

    \Comment{Receive neighbor trajectories.}
    $\Ncal^\com \leftarrow \set{\Tcal_{k_j,j}^\com}_{j \in \Ncal_i(t_{k_i + 1})}$\;
    
    \KwSty{success}, $\Tcal^{\can} \leftarrow \text{AttemptReplan}(x_i, g_i, \Ncal^\com, \Scal)$\;
    
    \If{\KwSty{success}}{
        $\Tcal_{k_i+1, i}^{\com} \leftarrow \Tcal^{\can}$ \Comment{Commit trajectory.}
        
        \If{$i \notin \Acal(t_{k_i + 1})$}{

            $\Acal(t_{k_i+1}) \leftarrow \Acal(t_{k_i+1}) \cup \{i\}$ \Comment{Join.}
        }
        \text{BroadcastNewCommitted}$(\Tcal_{k_i+1,i}^\com)$
    }
    \ElseIf{$i \in \Acal(t_{k_i + 1})$}{
        $\Tcal_{k_i+1, i}^{\com} \leftarrow \Tcal_{k_i, i}^{\com}$ \Comment{Keep old trajectory}
    }
\end{algorithm}

Finally, we combine these elements into a single theorem that proves the forward-invariant safety of the entire system.
This theorem asserts that as long as agents join the network and update their trajectories according to the specified rules, every agent will possess a valid committed trajectory at all times, thereby guaranteeing perpetual system-wide safety.

\begin{theorem}[Forward-Invariant Safety via Update Protocols]
\label{theorem:forward_invariant_safety}
Assume $\Acal(t_0) = \emptyset$. If all agents satisfy \eqref{eqn:r3r}, and join and replan by \cref{alg:update_state}, then for all $t \ge t_0$, the system remains safe:
\eqn{
    \norm{\p_i(t) - \p_j(t)}_2 \ge \delta, \quad &\forall i \neq j \in \Acal(t),\\
    \p_i(t) \in \Scal, \quad &\forall i \in \Acal(t).
}
\end{theorem}
\begin{proof}
    We prove safety by induction over update events (joining or replanning) generated by \cref{alg:update_state}.

    \emph{Base case.} At time $t_0$, the active set is empty, $\Acal(t_0)=\emptyset$, hence the safety conditions hold trivially. When the first agent $i$ joins at time $t_{0_i}$, \cref{alg:update_state} requires it to commit a trajectory $\Tcal_{0_i,i}^{\com}$ satisfying \cref{def:valid}. In particular, $\p_i(t)\in\Scal$ for all $t\ge t_{0_i}$, and collision avoidance is vacuous since no other agents are present.

    \emph{Inductive hypothesis.}
    Assume that at some update time $t_k$, every agent $i\in\Acal(t_k)$ is executing a committed trajectory $\Tcal_{k_i,i}^{\com}$ that was valid at $t_{k_i}$. Then, by \cref{theorem:no_update_safety}, the system is safe for all $t\ge t_k$.

    \emph{Inductive step.}
    Consider the next event at time $t_{k+1}>t_k$.

    \emph{Case 1: Agent joining.}
    If a new agent $i\notin\Acal(t_k)$ joins at $t_{k+1}$, then by
    \cref{alg:update_state} it may enter only after committing a valid trajectory $\Tcal_{0_i,i}^{\com}$ with respect to the currently committed trajectories of all $j \in \Ncal_i(t_{k+1})$, and safety follows from \cref{theorem:no_update_safety}.

    \emph{Case 2: Agent replanning.}
    If an agent $i\in\Acal(t_k)$ replans at $t_{k+1}$, then either:

    (i) \emph{AttemptReplan} fails, in which case agent $i$ continues executing its previous committed trajectory, and the system remains safe by the inductive hypothesis; or

    (ii) \emph{AttemptReplan} succeeds, in which case agent $i$ commits a new trajectory $\Tcal_{k_i+1,i}^{\com}$ that satisfies \cref{def:valid}. Thus, safety again follows from \cref{theorem:no_update_safety}.

    Thus, safety is preserved across every update event. Since $t_k\to\infty$, the system remains safe for all $t\ge t_0$.
\end{proof}
An example of a system evolving under this protocol is given in Table~\ref{fig:r3r_evolution}. We see by repeatedly replanning, agents are able to shift their anchor points and certify trajectories that enable them to safely progress towards their goal states.
\begin{remark}
    If an agent elects to replan at times $t_{k_i + 1} < t_{k_i} + \Delta t_S$ at each iteration, and is able to successfully construct a valid committed trajectory each time, it will never execute its backup controller. Therefore, an agent could avoid ever deviating from a nominal, goal-oriented trajectory, unless necessary to guarantee safety.
\end{remark}

\section{Simulation Results and Analysis}\label{sec:sims}

As a case study, we consider 2D curvature-constrained agents operating under the Dubins vehicle dynamics
\eqn{
\begin{bmatrix}
    \dot{x} \\
    \dot{y} \\
    \dot{\theta}
\end{bmatrix} = 
\begin{bmatrix}
    v \cos \theta \\
    v \sin \theta \\
    \omega
\end{bmatrix},
}
where $v$ is a fixed velocity, and $\omega$ is the angular velocity input, which is bounded by $|\omega| \leq \omega_\text{max}$ \cite{dubins_curves_1957}.
This system is representative of a fixed-winged aircraft operating at a constant velocity.
All simulations were performed in \texttt{Julia}, using Agents.jl \cite{datseris_agentsjl_2024}, running on a 2022 MacBook Air (Apple M2, 16 GB).

The nominal planner is a Dynamic-RRT*\footnote{Implementation: \url{https://github.com/MarshallVielmetti/DynamicRRT.jl}}, an RRT* planner based on Dubins motion primitives that plans dynamically feasible trajectories, while avoiding static obstacles and other agents.
Committed trajectories are constructed by maximizing the time for which the agent's nominal can be tracked before switching to the backup controller.
The backup controller for the Dubins dynamics was implemented as a time-parameterized loiter circle, the location of which is determined as part of constructing the candidate trajectory.

\begin{table}[h!]
    \centering
    \caption{Simulation results of R3R \& \gatekeeper{} approach.}
    \vspace{-0.15cm}
    \begin{tabular}{ccc}
        \hline
         \textbf{\#Agents - Environment} & \textbf{Safety} & \textbf{Success} \\
        \hline
         8-swap & 100\% & 100\% \\
         16-swap & 100\% & 100\% \\
         8-city-like & 100\% &  100\% \\
         16-city-like & 100\% & 100\% \\
         32-city-like & 100\% & 100\% \\
         64-city-like & 100\% & 100\% \\
         128-city-like & 100\% & 97\% \\
         8-willow-garage & 100\% & 100\% \\
         16-willow-garage & 100\% & 100\% \\
         \hline
    \end{tabular}
    \label{tab:safety_succ}
    \vspace{-0.15cm}
\end{table}

\begin{figure*}
    \centering
    \includegraphics[width=0.95\linewidth]{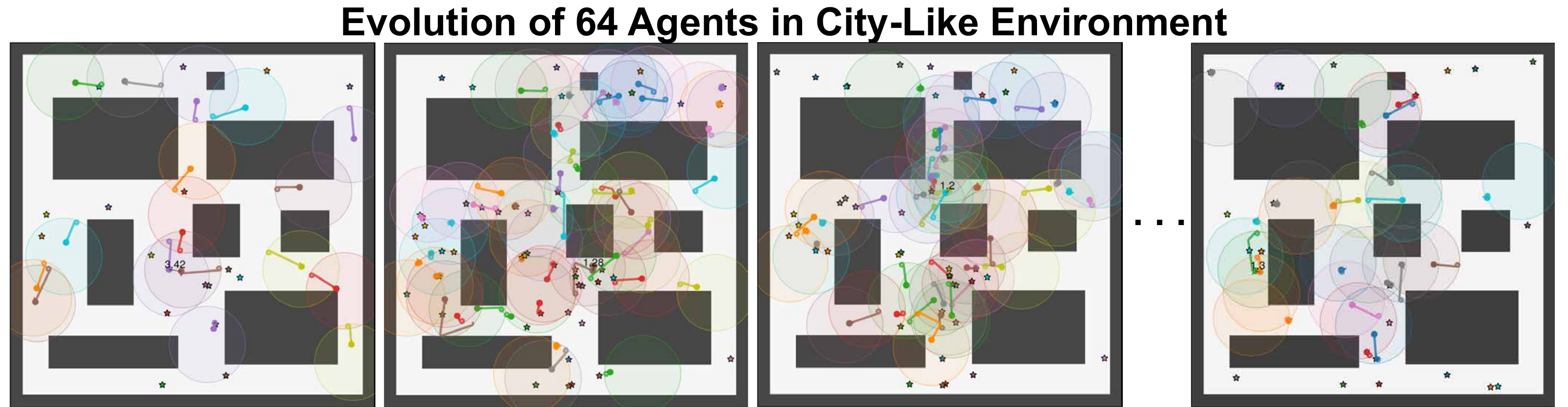}
    \caption{64 Dubins vehicles initialized with random start and goal positions in a city-like environment.}
    \label{fig:simulation_environment}
    \vspace{-0.5cm}
\end{figure*}

Results from the simulations are given in \cref{tab:safety_succ}.
Each scenario was run $5$ times.
For the city-like and willow garage scenarios, agent start and goal positions were randomized in each trial.
Safety is defined as all environmental constraints being satisfied, and no agents coming within inter-agent collision avoidance distance. 
Success was defined as the percent of agents making it to their goal positions.
The simulation parameters for the city-like environment (approximately $100\times100$) were: $\Rcomm = 16.0, \delta = 0.5, \Rplan = 5.16$, and maximum turning radius of $0.25$.  For the $32$-agent simulation, agents averaged $4.7$ neighbors at each replanning iteration, with a maximum of $10$. Replanning steps took an average of $1.4$ ms.

The simulation results validate our theoretical safety guarantees, and the high success rates demonstrate the algorithm's deadlock resilience, even in dense scenarios, an example of which is given in \cref{fig:simulation_environment}.
The only failure case occurred in the $128$-agent city-like scenario when a group of agents were unable to reach their goals within the simulation time limit, due to a combination of high density and limited free space. 
This resulted in a deadlock situation that the current implementation could not resolve. Safety was never violated.
Future work will focus on enhancing deadlock resolution strategies to further improve success rates in such challenging scenarios.

\subsection{Baseline Comparison}
While comparison against state-of-the-art (e.g., MADER~\cite{tordesillas2021mader}, DREAM~\cite{senbaslarDREAMDecentralizedRealTime2025}, and EGO-Swarm~\cite{zhouEGOSwarmFullyAutonomous2020}) is desirable, these methods fundamentally assume fully-connected communication graphs and global information sharing, and adapting them to operate under strict, distance-based communication constraints changes their core behavior. Relative to these methods, our approach sacrifices global optimality for guaranteed constraint satisfaction under partial information.

We compare the performance of our solver to a baseline Nonlinear MPC in Julia using Ipopt~\cite{wachter_implementation_2006} and JuMP~\cite{lubin_jump_2023}, which we demonstrate in an empty environment. Inter-agent collision constraints with respect to neighboring agents are encoded in the optimization problem as hard constraints. If the resulting problem is infeasible, we attempt to solve a relaxed version of the problem, with high penalties for constraint violation.

Due to the effects of the communication radius, the NMPC solver finds itself in unrecoverable states where no safe solution exists, and we see the rate of safety violations increase substantially as problem density increases.

\begin{table}[h!]
    \centering
    \caption{Evaluation of NMPC baseline.}
    \vspace{-0.15cm}
    \begin{tabular}{c|ccc}
        \hline
        \multirow{2}{*}{$n_{\text{agents}}$} & \multicolumn{3}{c}{\textbf{Safety} }\\ 
        & $\Rcomm = 1.0$ & $\Rcomm = 2.0$ & $\Rcomm = 10.0$\\ 
        \hline
        4  & 87.5\% & 100.0\% & 100.0\%\\
        6  & 87.5\% & 100.0\% & 100.0\%\\
        8  & 75.0\% & 87.5\%  & 100.0\%\\
        10 & 62.5\% & 87.5\% & 100.0\%\\
        14 & 50.0\% & 75.0\% & 100.0\%\\
        16 & 37.5\% & 37.5\% & 75.0\%\\ 
        \hline
    \end{tabular}
    \label{tab:nmpc_success}
    \vspace{-0.15cm}
\end{table}

\subsection{Runtime Analysis}
Let $N$ be the total number of agents, $A$ be the environment area, $\rho = N/A$ be the agent density, and $\Rcomm$ be the communication radius.
Assuming agents are distributed uniformly, at some replanning iteration, the average number of neighbors within communication range is $\lambda =  \rho \pi (\Rcomm)^2$.
All online costs depend on $\lambda$ (local density), not $N$, provided $\rho$ stays bounded as $N$ grows.
Thus, collision checking in nominal planning and candidate validation is $\Ocal(\lambda)$, not $\Ocal(N)$ as with centralized or fully-connected approaches.
This analysis is confirmed empirically by the results given in \cref{fig:planner_performance}.
\begin{figure}
    \centering
    \includegraphics[width=0.9\linewidth]{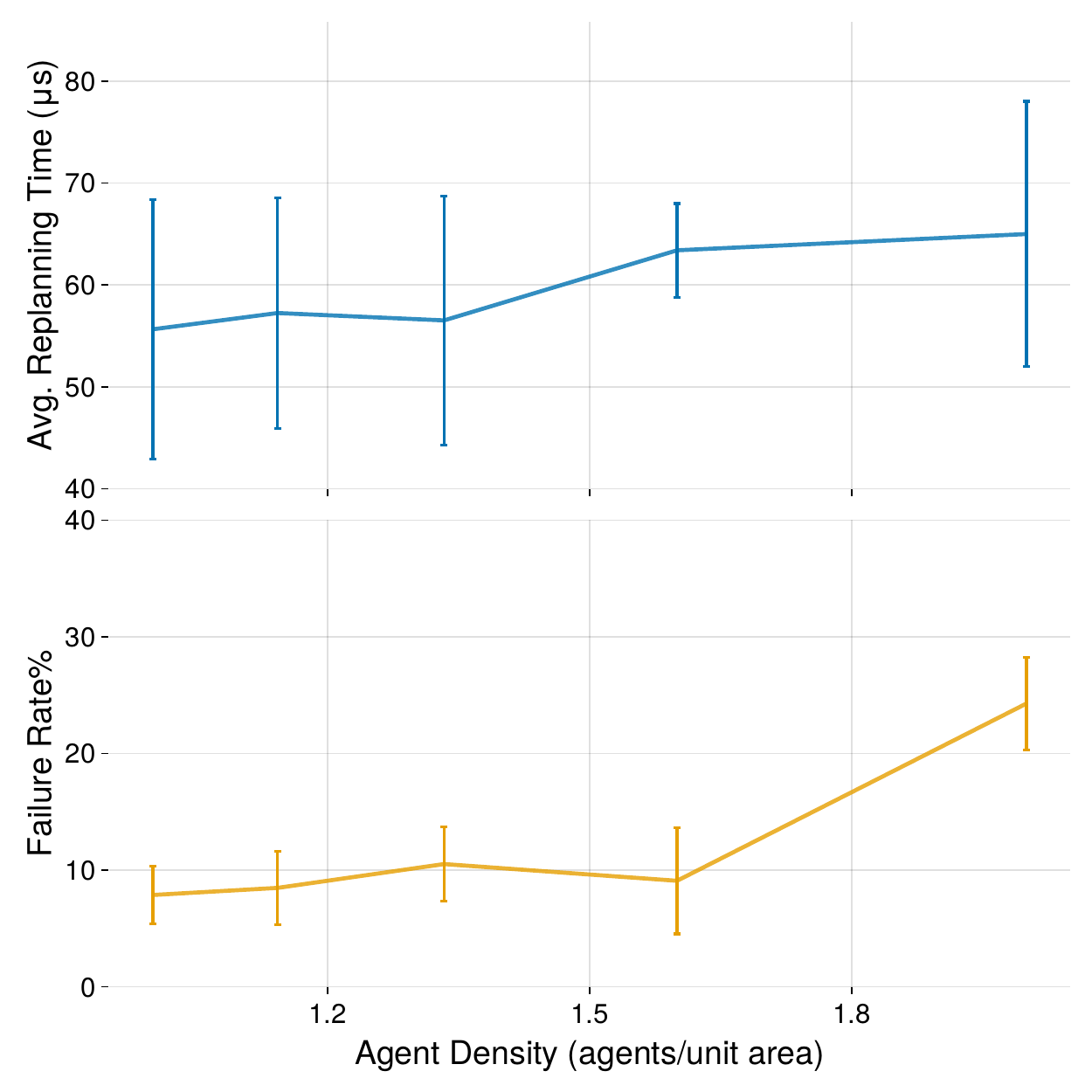}
    \vspace{-0.4cm}
    \caption{Runtime performance of R3R as a function of agent density for $N$ agents in a variable-sized environment over 20 randomized trials. \emph{(Top)} The average time per replanning attempt increases with density due to a higher number of neighboring agents to consider in collision checks. \emph{(Bottom)} The rate of replanning failures (no valid candidates found) also increases with density as it becomes more difficult to find a valid trajectory. }
    \label{fig:planner_performance}
\end{figure}

\vspace{-1.0em}
\section{Conclusion}
\label{sec:conclusions}
This paper presented R3R, a novel framework that establishes a formal, geometric link between an agent’s communication range and the bounds within which it can plan its trajectories.
Specifically, we introduced $R$-Bounded trajectories, demonstrating that restricting an agent's plan to a function of its communication radius and collision avoidance radius geometrically guarantees that local deconfliction implies global safety.
We leverage \gatekeeper{} to ensure that all agents asynchronously construct and track valid, $R$-Bounded trajectories, achieving forward-invariant safety without global information.
Our work demonstrates that for a broad class of systems, scalability and formal safety guarantees are not mutually exclusive goals but can be achieved simultaneously through principled geometric constraints.
The efficacy of the algorithm is demonstrated via the asynchronous deconfliction of >100 agents in obstacle-rich environments.
Our future work will focus on addressing the method's limitations, namely to enhance the framework's robustness, resolve deadlocks, and incorporate models of network uncertainty to provide probabilistic safety guarantees under realistic communication conditions.



\vspace{-0.3cm}
\bibliographystyle{IEEEtran}
\bibliography{biblio}

\end{document}

%% file: preamble/preamble.tex
\input{preamble/packages}
\input{preamble/commands}

\input{preamble/acronyms}

\input{preamble/theorems}

%% file: preamble/packages.tex
\usepackage{acronym}

\usepackage{blindtext} 

\usepackage[T1]{fontenc}

\usepackage{amsmath}
\usepackage{amssymb}
\usepackage{amsfonts}
\usepackage{graphicx}
\usepackage{textcomp}
\usepackage{setspace}
\usepackage{booktabs}
\usepackage{multirow}
\usepackage[table]{xcolor} 
\usepackage{epsfig}
\usepackage{svg}
\usepackage[ruled,vlined]{algorithm2e} 
\usepackage{bm}
\usepackage{tabularx}
\usepackage{placeins} 
\usepackage{siunitx} 

\usepackage{xurl}



%% file: preamble/commands.tex
\newcommand{\naturals}{\mathbb{N}}

\newcommand{\reals}{\mathbb{R}}
\newcommand{\R}{\reals}
\newcommand{\Rnonneg}{\reals_{\geq 0}}
\newcommand{\Rplus}{\reals_{>0}}

\renewcommand{\emptyset}{\varnothing}

\newcommand{\norm}[1]{\left\Vert #1 \right \Vert}

\newcommand{\bmat}[1]{\begin{bmatrix}#1\end{bmatrix}}

\newcommand{\Acal}{\mathcal{A}}
\newcommand{\Bcal}{\mathcal{B}}
\newcommand{\Ccal}{\mathcal{C}}

\newcommand{\Gcal}{\mathcal{G}}

\newcommand{\Ical}{\mathcal{I}}

\newcommand{\Ncal}{\mathcal{N}}
\newcommand{\Ocal}{\mathcal{O}}

\newcommand{\Scal}{\mathcal{S}}
\newcommand{\Tcal}{\mathcal{T}}
\newcommand{\Ucal}{\mathcal{U}}

\newcommand{\Wcal}{\mathcal{W}}
\newcommand{\Xcal}{\mathcal{X}}

\newcommand{\eqn}[1]{\begin{align} #1 \end{align}}
\newcommand{\neqn}[1]{\begin{align*} #1 \end{align*}}

\newcommand{\set}[1]{\{#1\}}

\usepackage[normalem]{ulem}

\newcommand{\gatekeeper}{\texttt{gatekeeper}}

\newcommand{\nom}{{\rm nom}}
\newcommand{\can}{{\rm can}}
\newcommand{\com}{{\rm com}}
\newcommand{\back}{{\rm bak}}

\newcommand{\Rplan}{R^\text{plan}}
\newcommand{\Rcomm}{R^\text{comm}}
\newcommand{\x}{{\rm x}}
\renewcommand{\u}{{\rm u}}
\newcommand{\p}{{\rm p}}

\newcolumntype{g}{>{\columncolor{gray!30}}r}



%% file: preamble/acronyms.tex
\acrodef{BCH}[BCH]{Baker-Campbell-Hausdorff}
\acrodef{CBF}[CBF]{Control Barrier Function}
\acrodef{CBF-QP}[CBF-QP]{Control Barrier Function Quadratic Program}
\acrodef{CDC}[CDC]{Conference on Decision and Control}
\acrodef{CESDF}[CESDF]{Certified ESDF}
\acrodef{CLF}[CLF]{Control Lyapunov Function}
\acrodef{CVO}[C-VO]{Certified Visual Odometry}
\acrodef{DCT}[DCT]{Discrete Cosine Transform}
\acrodef{DMP}[DMP]{Distance Map Planner}
\acrodef{EKF}[EKF]{Extended Kalman Filter}
\acrodef{ESDF}[ESDF]{Euclidean Signed Distance Field}
\acrodef{EZ}[EZ]{Engagement Zone}
\acrodef{FOV}[FoV]{Field of View}
\acrodef{FPV}[FPV]{First Person View}
\acrodef{GNC}[GNC]{Graduated-Nonconvexity}
\acrodef{GP}[GP]{Gaussian Process}
\acrodef{HOCBF}[HOCBF]{Higher Order Control Barrier Function}
\acrodef{ICCBF}[ICCBF]{Input-Constrained Control Barrier Function}
\acrodef{IEEE}[IEEE]{Institute of Electrical and Electronics Engineers}
\acrodef{IMU}[IMU]{Inertial Measurement Unit}
\acrodef{ISS}[ISS]{Input-to-State}
\acrodef{KF}[KF]{Kalman Filter}
\acrodef{ML}[ML]{Machine Learning}
\acrodef{MPC}{Model Predictive Control}
\acrodef{NGPKF}[NGPKF]{Numerical Gaussian Process Kalman Filter}
\acrodef{ODE}[ODE]{Ordinary Differential Equation}
\acrodef{QP}[QP]{Quadratic Program}
\acrodef{RGBD}[RGBD]{RGB-Depth}
\acrodef{RL}[RL]{Reinforcement Learning}
\acrodef{RoS}[RoS]{Rate of Spread}
\acrodef{SDE}[SDE]{Stochastic Differential Equation}
\acrodef{SDF}[SDF]{Signed Distance Field}
\acrodef{SFC}[SFC]{Safe Flight Corridor}
\acrodef{SLAM}[SLAM]{Simultaneous Localization and Mapping}
\acrodef{SOS}[SOS]{Sum of Squares}
\acrodef{SVD}[SVD]{Singular Value Decomposition}
\acrodef{TCAC}[TCAC]{Technical Committee on Aerospace Controls}
\acrodef{TLS}[TLS]{Truncated Least Squares}
\acrodef{TSDF}[TSDF]{Truncated Signed Distance Field}
\acrodef{TSD}[TSD]{Target Spatial Distribution}
\acrodef{VIO}[VIO]{Visual Inertial Odometry}
\acrodef{VO}[VO]{Visual Odometry}
\acrodef{WLS}[WLS]{Weighted Least Squares}

%% file: preamble/theorems.tex
\usepackage{amsthm}

\theoremstyle{plain}
\newtheorem{theorem}{Theorem}

\newtheorem{lemma}{Lemma}
\newtheorem{problem}{Problem}
\newtheorem{definition}{Definition}

\newtheorem{assumption}{Assumption}
\newtheorem{remark}{Remark}

\theoremstyle{remark}

\AtBeginEnvironment{definition}{\let\oldemph\emph \renewcommand{\emph}[1]{\textbf{#1}}}
\AtEndEnvironment{definition}{\let\emph\oldemph}

\usepackage{cleveref}
\crefformat{problem}{Problem~#2#1#3}
\crefformat{assumption}{Assumption~#2#1#3}

%% file: sections/1_introduction.tex
From last mile delivery~\cite{eskandaripourLastMileDroneDelivery2023} to warehouse robotics~\cite{wenSwarmRoboticsControl2018} to autonomous drone swarms~\cite{javedStateoftheArtFutureResearch2024}, deploying multi-agent systems safely in the real world is bottlenecked by a lack of formal safety guarantees.
This challenge, known as multi-agent motion planning (MAMP), forces a trade-off between guarantees, coordination, and scalability. 
Centralized planners can provide strong safety guarantees but fail to scale and are prone to single-node failure. 
Conversely, decentralized methods scale well to large teams but have thus far been unable to provide formal, infinite-horizon safety guarantees under distance-based communication constraints for non-trivial (i.e., beyond single/double integrator) systems.

Multi-agent path finding (MAPF), a discrete state precursor to MAMP, is typically formulated as a graph search problem. 
Classical MAPF methods, such as Conflict-Based Search and its variants~\cite{sharonConflictbasedSearchOptimal2015,boyarski_icbs_2021,gange_lazy_2019,liSymmetryBreakingConstraintsGridBased2019}, provide strong optimality or bounded suboptimality guarantees, but these approaches are centralized and limited to discrete domains and thus may not produce dynamically feasible solutions.
Some MAMP methods have extended these concepts to continuous space using trajectory optimization or sampling-based planning, but these centralized approaches struggle with scalability and real-time performance~\cite{capMultiagentRRTSamplingbased2013}.

Planning-based methods comprise the majority of the recent state-of-the-art. 
These methods guarantee safety by generating safe trajectories over some finite horizon. 
This also allows for agents to take preemptive actions, reducing deadlocks and allowing for more direct paths compared to reactive methods.
Despite being distributed in computation, these methods are often centralized in information, or use synchronization to enable coordination.
In MADER~\cite{tordesillas2021mader,kondoRobustMADERDecentralized2023}, agents construct trajectories individually, but must share and receive full trajectories with all other agents in the network.
Other methods require synchronous replanning/state updates \cite{parkOnlineDistributedTrajectory2022,parkDLSCDistributedMultiAgent2023,leeMCSwarmMinimalCommunicationMultiAgent2025,parkDecentralizedTrajectoryPlanning2025} and full state sharing among all agents in the network \cite{zhouEGOSwarmFullyAutonomous2020,zhou_swarm_2022,senbaslarDREAMDecentralizedRealTime2025}.
While these methods have been shown to outperform reactionary methods, they still provide only finite-horizon guarantees.

Critically, no existing decentralized, asynchronous approach guarantees infinite-horizon safety under distance-based communication constraints. 
This gap motivates our proposed framework, R3R. 
To the best of our knowledge, R3R represents the first decentralized and asynchronous framework to provide infinite-horizon safety guarantees for multi-agent motion planning under distance-based communication constraints which generalizes to broad classes of systems under certain assumptions. 
This is achieved by integrating a novel geometric constraint, which we term R-Boundedness, with \gatekeeper{}~\cite{agrawalGatekeeperOnlineSafety2024}, our safety framework developed for single agents, to provide recursive feasibility guarantees in multi-agent systems.
Our contributions are as follows:
\begin{enumerate}
    \item \emph{R3R}, a provably safe, decentralized, and asynchronous framework for multi-agent motion planning under communication constraints. 
    \item \emph{R-Boundedness}, a novel geometric constraint on trajectory generation, which we show links an agent's ability to communicate with its ability to plan safely.
    \item We demonstrate the efficacy of our method in \emph{simulations of up to 128 agents}, showcasing safe and scalable multi-agent coordination.
\end{enumerate}

The paper is organized as follows: Section \ref{sec:prelims} introduces preliminaries and the problem statement, section \ref{sec:approach} presents the proposed approach, the methodology is evaluated in simulation in section \ref{sec:sims}, and conclusions and thoughts for future work are presented in section \ref{sec:conclusions}.